\def\useieeelayout{0}
\def\showall{0}

\newcommand{\inConf}[1]{\if\useieeelayout1{#1}\fi\if\showall1{\color{green!50!black}In Conf: #1}\fi}
\newcommand{\inArxiv}[1]{\if\useieeelayout0{#1}\else\if\showall1{\color{blue}In ArXiV: #1}\fi\fi}

\if\useieeelayout1
\documentclass[letterpaper, 10 pt, conference]{ieeeconf}  
\IEEEoverridecommandlockouts                              %
\usepackage{graphicx}
\usepackage{amsmath}
\usepackage{optidef}
 
\usepackage{amsthm}
\usepackage{amssymb}
\usepackage{multicol}
\usepackage{enumitem}
\usepackage{cuted}
\usepackage{multirow}
\usepackage{caption, subcaption}
\newtheorem{theorem}{Theorem}
\newtheorem{lemma}{Lemma}
\newtheorem{corollary}{Corollary}
\newtheorem{proposition}{Proposition}
\newtheorem{definition}{Definition}
\newtheorem{problem}{Problem}
\newtheorem{assumption}{Assumption}
\usepackage{xcolor}
\usepackage{cite}
\usepackage[colorlinks=true,allcolors=steelblue]{hyperref}
\definecolor{mycolor1}{rgb}{0.00000,0.44700,0.74100}%
\definecolor{mycolor2}{rgb}{0.85000,0.32500,0.09800}%

\usepackage{comment}
\usepackage{algorithm}
\usepackage{algorithmic}
\usepackage{scalerel}
\usepackage{bm}
\usepackage{pgfplots}
\usepackage{graphicx}
\usepackage{stmaryrd}
\usepackage{tikz}
\usetikzlibrary{plotmarks}
\pgfplotsset{compat=newest}
\pgfplotsset{plot coordinates/math parser=false}
\usepackage{grffile}
\pgfplotsset{compat=newest}
\usetikzlibrary{plotmarks}
\usetikzlibrary{arrows.meta}
\usepgfplotslibrary{patchplots}
\definecolor{steelblue}{RGB}{70,130,180}
\definecolor{mycolor1}{rgb}{0.00000,0.44700,0.74100}%

\def\bbn{\mathbb N}
\def\bbz{\mathbb Z}
\def\bbr{\mathbb R}

\DeclareMathOperator*{\argmin}{argmin}

\newcommand{\norm}[1]{\left\|#1 \right\|}
\newcommand{\abs}[1]{\left |#1 \right\vert}

\title{\LARGE \bf
Distributed Adaptive Control For Uncertain Networks}

\author{Venkatraman Renganathan, Anders Rantzer, and Olle Kjellqvist 
\thanks{
This project has received funding from the European Research Council (ERC) under the European Union’s Horizon 2020 research and innovation program under grant agreement No 834142 (Scalable Control). The authors are with the Department of Automatic Control LTH, Lund University, Sweden and they are members of the ELLIIT Strategic Research Area. (E-mail: (venkatraman.renganathan,anders.rantzer, olle.kjellqvist)@control.lth.se).
}
}

\else
\documentclass[10pt]{article}
\usepackage{graphicx}
\usepackage{amsmath}
\usepackage{optidef}
 
\usepackage{amsthm}
\usepackage{amssymb}
\usepackage{amsfonts}       
\usepackage{multicol}
\usepackage{cuted}
\usepackage{multirow}
\usepackage{caption, subcaption}
\newtheorem{theorem}{Theorem}
\newtheorem{lemma}{Lemma}
\newtheorem{proposition}{Proposition}

\newtheorem{problem}{Problem}
\newtheorem{assumption}{Assumption}

\usepackage{xcolor}
\usepackage{cite}
\usepackage[colorlinks=true,allcolors=steelblue]{hyperref}
\definecolor{mycolor1}{rgb}{0.00000,0.44700,0.74100}%
\definecolor{mycolor2}{rgb}{0.85000,0.32500,0.09800}%

\usepackage{comment}
\usepackage{algorithm}
\usepackage{algorithmic}
\usepackage{scalerel}
\usepackage{bm}
\usepackage{url}
\usepackage{pgfplots}
\usepackage{graphicx}
\usepackage{stmaryrd}
\usepackage{tikz}
\usetikzlibrary{plotmarks}
\pgfplotsset{compat=newest}
\pgfplotsset{plot coordinates/math parser=false}
\usepackage{grffile}
\pgfplotsset{compat=newest}
\usetikzlibrary{plotmarks}
\usetikzlibrary{arrows.meta}
\usepgfplotslibrary{patchplots}
\definecolor{steelblue}{RGB}{70,130,180}
\definecolor{mycolor1}{rgb}{0.00000,0.44700,0.74100}%

\usepackage[preprint]{tmlr}

\title{Distributed Adaptive Control For Uncertain Networks}

\author{\name Venkatraman Renganathan  \email       venkatraman.renganathan@control.lth.se \\
      \addr Department of Automatic Control - LTH \\ Lund University, Sweden
      \AND
      \name Anders Rantzer \email anders.rantzer@control.lth.se \\
      \addr Department of Automatic Control - LTH \\ Lund University, Sweden
      \AND
      \name Olle Kjellqvist \email olle.kjellqvist@control.lth.se \\
      \addr Department of Automatic Control - LTH \\ Lund University, Sweden
      }

\def\month{09}  
\def\year{2023} 
\def\openreview{\url{https://openreview.net/forum?id=XXXX}} 
\fi
\begin{document}

\maketitle
\thispagestyle{empty}
\pagestyle{empty}

\begin{abstract}
Control of network systems with uncertain local dynamics has remained an open problem for a long time. In this paper, a distributed minimax adaptive control algorithm is proposed for such networks whose local dynamics has an uncertain parameter possibly taking finite number of values. To hedge against this uncertainty, each node in the network collects the historical data of its neighboring nodes to decide its control action along its edges by finding the parameter that best describes the observed disturbance trajectory. Our proposed distributed adaptive controller is scalable and we give both lower and upper bounds for its $\ell_{2}$ gain. Numerical simulations demonstrate that once each node has sufficiently estimated its local uncertainty, the distributed minimax adaptive controller behaves like the optimal distributed $\mathcal{H}_{\infty}$ controller in hindsight. 

\end{abstract}

\section{Introduction}
\label{sec:introduction}
Control of large-scale and complex systems is often performed in a distributed manner \cite{antonelli2013interconnected}, as it is practically difficult for every agent in the network to have access to the global information about the overall networked system while deciding its control actions. On the other hand, designing optimal distributed control laws when the networked system dynamics are uncertain still remains an open problem. This naturally calls for a learning-based controller to be employed in such uncertain settings. Learning based controllers for network systems is very much in its infancy and recently a scalable solution was proposed in \cite{nali_network_learning}. Adaptive control in the centralised setting has been investigated a lot starting from \cite{astromwittenmark}, where an adaptive controller was shown to learn the system dynamics online through sufficient parameter estimation and then control it. Multiple model-based adaptive control formulation have been 
known to handle uncertainty in system dynamics and an extensive literature in that topic can be found in \cite{anderson2000multiple, hespanha2001multiple, anderson2001multiple, kuipers2010multiple, narendra1997adaptive, narendra2011changing}. Another promising approach was introduced in \cite{didinsky1994minimax} with the minimax problem formulation where the resulting full information problem of higher dimension was solved using Dynamic Programming. Minimax adaptive control formulation was specialized to linear systems with unknown sign for state matrix in \cite{rantzer2020minimax}, and to finite sets of linear systems in \cite{rantzer2021minimax, cederberg2022synthesis, renganathan2023online}. In general, minimax adaptive control problems are challenging mainly due to the exploration and exploitation trade-off that inevitably comes with the learning and the controlling procedure. 

On the other hand, designing optimal distributed control laws that address the uncertainty prevailing over true model of the networked system still remains an open problem. Though we can approach this problem from the multiple model-based adaptive control techniques, the resulting controller does not facilitate a distributed implementation as the controller solutions are often dense in nature. Aiming for a distributed implementation adds an additional layer of complexity to the existing challenges of any centralised adaptive control algorithm. However, there are certain classes of systems for which scalable implementation of distributed minimax adaptive control is possible, such as systems with inherent structure that open up the door for optimal control laws to be structured as well. Such system models are common in many infrastructural networks such as irrigation and transportation networks. 

Control of spatially invariant systems such as linear models of transportation and buffer networks in \cite{bamieh2002distributed} has paved the ways for designing distributed robust controllers for such special class of systems. We consider similar class of systems in our research problem where the original system comprises of subsystems with local dynamics, that only share control inputs. A closed-form expression for the distributed $\mathcal{H}_{\infty}$ optimal state feedback law for systems with symmetric and Schur state matrix was computed in \cite{lidstrom2017h}, where the total networked system comprised of subsystems with local dynamics, that share only control inputs and each control input affecting only two subsystems. Similarly, a closed-form expression for a decentralised $\mathcal{H}_{\infty}$ optimal controller with diagonal gain matrix for network systems having acyclic graphs was computed in \cite{vladu2022decentralized}. In all these previous works, the network dynamics are known exactly. 

\textit{Contributions:} We extend the problem setting in \cite{lidstrom2017h} by considering finite number of possible local dynamics in each node and control action along each edge. The highlight of our work is that we propose learning in network systems for addressing uncertainty in local dynamics along with disturbance rejection. Our main contributions are as follows:
\begin{enumerate}
    \item A scalable \& distributed minimax adaptive control algorithm for uncertain networked systems is developed. Each node in the network hedges against the uncertainty in its local dynamics by maintaining the history of just its neighboring nodes and finds the controller at any time by choosing the model that best describes the local disturbance trajectory (See equation \eqref{eqn_distributed_minimax_control}).
    \item Both lower and upper bounds for the $\ell_{2}$ gain associated with the proposed distributed minimax adaptive control algorithm are given (See Lemma \ref{lemma_l2_gain_lower_bound} \& Theorem \ref{theorem_2}).
    \item The efficacy of the proposed distributed minimax adaptive control is demonstrated using a large-scale buffer network with $10^{4}$ nodes (where computing a centralized controller is costly) where controller implementation does not require the knowledge about the $\ell_{2}$ gain.
\end{enumerate}
Following a short summary of notations, this paper is organized as follows: In \S\ref{sec:problemformulation}, the main problem formulation of distributed minimax adaptive controller is presented. The proposed distributed implementation of the minimax adaptive control algorithm along with the computation of the lower and upper bounds for its $\ell_{2}$ gain are given in \S\ref{sec:analysis}. The proposed algorithm is then demonstrated in \S\ref{sec:simulation}. Finally, the main findings of the paper are summarised in \S\ref{sec:conclusion} along with some directions for future research.

\section*{Notations}
The cardinality of the set $A$ is denoted by $\left | A \right \vert$. The set of real numbers, integers and the natural numbers are denoted by $\bbr, \bbz, \bbn$ respectively. For $N \in \mathbb{N}$, we denote by $\llbracket N \rrbracket := \{1,\dots,N \}$.
A vector of size $n$ with all values being one is denoted by $\mathbf{1}_{n}$. For matrix $A \in \bbr^{n \times n}$, we denote all its eigenvalues, transpose and its trace by $\mathrm{eig}(A)$, $A^{\top}$ and $\mathbf{Tr}(A)$ respectively. We denote by $\mathbb{S}^{n}$ the set of symmetric matrices in $\bbr^{n \times n}$ and the set of positive definite and positive semi-definite matrices by $\mathbb{S}^{n}_{++}$ and $\mathbb{S}^{n}_{+}$ respectively. A symmetric matrix $P \in \mathbb{S}^{n}$ is called positive definite (positive semi-definite) if $\forall x \in \bbr^{n} \backslash \{0\}, x^{\top} P x > 0 \, \, (x^{\top} P x \geq 0)$ and is denoted by $P \succ 0 (P \succeq 0)$. An identity matrix in dimension $n$ is denoted by $I_{n}$. Given $x \in \bbr^{n}, A \in \bbr^{n \times n}, B \in \bbr^{n \times n}$, the notations $\abs{x}^{2}_{A}$ and ${\left \| B \right \Vert}^{2}_{A}$ mean $x^{\top} A x$ and $\mathbf{Tr}\left(B^{\top} A B \right)$ respectively. 

\section{Problem Formulation}
\label{sec:problemformulation}
Control of spatially invariant systems has a rich literature in control theory (see \cite{bamieh2002distributed} and the references therein) where networked systems such as linear models of transportation and buffer networks are studied in great detail. We consider similar class of systems where the original system comprises of subsystems with local dynamics, that only share control inputs. Note that such systems can be naturally associated with a graph and hence we depict the subsystems as nodes and control inputs as edges between the nodes they affect. 

\subsection{Network Model}
We model such networked dynamical system in discrete time setting with a graph $\mathcal{G}$ comprising a node set $\mathcal{V}$ representing $|\mathcal{V}| = N$ subsystems and edge set $\mathcal{E} \subset \mathcal{V} \times \mathcal{V}$ representing a set of $|\mathcal{E}| = E$ communication links amongst the subsystems. The incidence matrix encoding the edge set information is denoted by $\mathcal{I} \in \mathbb{R}^{N \times E}$. We denote by $\mathcal{N}_{i} = \{j \in \mathcal{V} : (j, i) \in \mathcal{E}\}$, the neighbor set of agent $i$, whose states are available to agent $i$ through $\mathcal{E}$. The degree of node $i \in \mathcal{V}$ is denoted as $d_{i} := \abs{\mathcal{N}_{i}}$. The set of inclusive neighbors of agent $i$ is denoted by $\mathcal{J}_{i} := \{i\} \cup \mathcal{N}_{i}$. Let $\mathbf{d} := \max\{ d_{i} \mid i \in \mathcal{V} \}$. We associate with each node $i \in \mathcal{V}$, a state $x_{i}(t) \in \mathbb{R}$ at time $t \in \mathbb{N}$. Each subsystem $\mathbf{\Sigma_{i}}$ corresponding to node $i \in \mathcal{V}$ updates its own states by interacting with its neighbors as 
\begin{align} \label{eqn_node_dynamics}
    x_{i}(t+1) = a_{i} x_{i}(t) + b \sum_{j \in \mathcal{N}_{i}} (u_{i}(t) - u_{j}(t))  + w_{i}(t),
\end{align}
where $a_{i} \in (0,1), b > 0$. The additive disturbance $w_{i}(t) \in \mathbb{R}$ affecting the node $i \in \mathcal{V}$ is adversarial in nature. The control input between two nodes $(i,j) \in \mathcal{E}$ is such that what is drawn from subsystem $j$ is added to subsystem $i$ and hence it only affects two nodes. Note that the dynamics of each node in the network is coupled with the other nodes only through their control inputs. The concatenated states of all nodes in \eqref{eqn_node_dynamics}, the corresponding control inputs along the edges and the adversarial disturbances acting on each node are denoted as
\begin{align}
x(t) &= \begin{bmatrix} x_{1}(t) & \cdots & x_{N}(t) \end{bmatrix}^{\top} \in \mathbb{R}^{N}, \\
u(t) &= \{u_{i}(t)  - u_{j}(t)\}_{(i,j) \in \mathcal{E}} \in \mathbb{R}^{E}, \text{ and }\\
w(t) &= \begin{bmatrix} w_{1}(t) & \cdots & w_{N}(t) \end{bmatrix}^{\top} \in \mathbb{R}^{N}    
\end{align}
respectively. Hence, the system described by \eqref{eqn_node_dynamics} can be equivalently written compactly as 
\begin{align} \label{eqn_compact_dynamics}
x_{t+1} = A x_{t} + B u_{t} + w_{t},     
\end{align}
with $A \in \mathbb{R}^{N \times N}$ being symmetric and Schur stable, and $B = b \mathcal{I}$, with $\mathcal{I} \in \mathbb{R}^{N \times E}$ being the incidence matrix of the underlying graph. Note that, $BB^{\top} = b^{2} \mathcal{I} \mathcal{I}^{\top}= b^{2}\mathbf{L}$, where $\mathbf{L}$ denotes the Laplacian matrix associated with the graph $\mathcal{G}$.

\subsection{An Optimal Distributed $\mathcal{H}_{\infty}$ Controller}
For systems described by \eqref{eqn_node_dynamics}, it was shown in \cite{lidstrom2017h} that an optimal distributed $\mathcal{H}_{\infty}$ controller is given by 
\begin{align} \label{eqn_opt_distributed_Hinf_control}
    K = B^{\top} (A-I)^{-1},
\end{align}
as long as the dynamics of each node $i \in \mathcal{V}$ satisfies the following condition 
\begin{align} \label{eqn_communication_condition}
a^{2}_{i} + 2 b^{2} d_{i} < a_{i}.
\end{align} 
The condition \eqref{eqn_communication_condition} is related to the speed of information propagation through the network as well as its connectivity defined using the parameters $a_i, b, d_{i}$, and the bound on the maximum eigenvalue of the symmetric normalized Laplacian matrix of the underlying network's graph. More details on the local condition of nodes described by \eqref{eqn_communication_condition} is available in subsection IV.B of \cite{lidstrom2017h}. Further, \eqref{eqn_communication_condition} can be equivalently written using the compact notations as 
\begin{align} \label{eqn_compact_communication_constraint}
A^{2} + BB^{\top} \prec A.    
\end{align} 
Since $a_{i} \in (0,1)$, \eqref{eqn_node_dynamics} is inherently stable. Then, it is best to quantify the amplification caused by just the disturbance to the system when a zero control is applied. We have the following lemma that computes the $\ell_{2}$ gain of the subsystem $\mathbf{\Sigma_{i}}$ with zero control inputs.
\begin{lemma}
Let $u_{i}(t) = u_{j}(t)$ for all $j \in \mathcal{N}_{i}$ corresponding to node $i \in \mathcal{V}$ in \eqref{eqn_node_dynamics}. Then, 
\begin{align} \label{eqn_ell_2_gain_zero_controls}
&\norm{\mathbf{\Sigma_{i}}}_{\mathcal{H}_{\infty}} 
\leq 
\frac{1}{1 - \left( \frac{1}{2} + \sqrt{\frac{1}{4} - 2 b^{2} d_{i} } \right) }, \quad \text{if}, \\
&\frac{1}{2} - \sqrt{\frac{1}{4} - 2 b^{2} d_{i}} < a_{i} < \frac{1}{2} + \sqrt{\frac{1}{4} - 2 b^{2} d_{i}}, \quad \text{and} \label{eqn_bounds_ai} \\
&b < \sqrt{\frac{1}{8 \mathbf{d}}}. \label{eqn_b_condition}    
\end{align}
\end{lemma}
\begin{proof}
We can equivalently rewrite \eqref{eqn_communication_condition} as
\begin{align*} 
&\left(a_{i} - \frac{1}{2} \right)^{2} + 2 b^{2} d_{i} - \frac{1}{4} < 0 \\
\iff
&\frac{1}{2} - \sqrt{\frac{1}{4} - 2 b^{2} d_{i}} < a_{i} < \frac{1}{2} + \sqrt{\frac{1}{4} - 2 b^{2} d_{i}}.
\end{align*}
Since $a_{i} \in (0,1)$ and $d_{i} \geq 1, \forall i \in \mathcal{V}$, we see that 
\begin{align*}
\frac{1}{4} - 2 b^{2} d_{i} > 0
\iff
b < \sqrt{\frac{1}{8 d_{i}}}.    
\end{align*}
Since the above condition is true $\forall i \in \mathcal{V}$, we get \eqref{eqn_b_condition}. Note that \eqref{eqn_node_dynamics} relates to a first-order stable system with zero control input for which its $\ell_{2}$ gain is given by
\begin{align*} 
\norm{\mathbf{\Sigma_{i}}}_{\mathcal{H}_{\infty}} 
\leq 
\frac{1}{1 - \left( \frac{1}{2} + \sqrt{\frac{1}{4} - 2 b^{2} d_{i} } \right) }.
\end{align*}
\end{proof}
\subsection{Distributed Minimax Adaptive Control Problem}
In the problem setting considered in this paper, the true system model $a_{i}$ in \eqref{eqn_node_dynamics} governing the dynamics of each node $i \in \mathcal{V}$ is \emph{unknown}. We have the following two assumptions to characterize the uncertainty in the local dynamics of each node in the network. 
\begin{assumption} \label{assume_1}
The parameter $a_{i}$ in \eqref{eqn_node_dynamics} belongs to a finite set $\mathbf{A_{i}}$ with $\abs{\mathbf{A_{i}}} = M \in \mathbb{N}_{\geq 1}$ elements. That is, $\forall i \in \mathcal{V}$,
\begin{equation} \label{eqn_node_uncertainty}
    a_{i} \in \mathbf{A_{i}} := \{ a_{i} \in (0,1) \mid \eqref{eqn_communication_condition} \text{ is satisfied} \}.
\end{equation}    
\end{assumption}

\begin{assumption} \label{assume_2}
The dynamics of node $i \in \mathcal{V}$ is independent of the dynamics of its neighbor $j \in \mathcal{N}_{i}$. This means that the choice of any one of the $M$ values of $a_i \in \mathbf{A_{i}}$ being the true $a_{i}$ does not influence any one of the $M$ values of $a_j \in \mathbf{A_{j}}$ being the true $a_{j}$ for every neighbor $j \in \mathcal{N}_{i}$. Hence, there are $F := M^{N}$ possible realisations of system matrix $A$. 
\end{assumption}

Given that the local dynamics of each node in the network is being uncertain, the controller should first learn the local dynamics accurately and then guarantee robustness against the adversarial disturbance acting on the node. This clearly calls for a learning-based control policy which aims to learn the uncertain parameter through the collected history of system data. However, we would also need to ensure that the learning based control policy does not result in a $\ell_{2}$ gain for the system that is worse than the one with zero control input given by \eqref{eqn_ell_2_gain_zero_controls}. A control policy is termed as \emph{admissible} if it is stabilising and has causal implementation. The control input of node $i \in \mathcal{V}$ depends upon the historical data of only 
its inclusive neighbors and is given by
\begin{align} \label{eqn_control_policy}
    u_{i}(t) 
    = \pi_{t}\left( \left\{ x_{j}(\tau) \right\}^{t}_{\tau = 0}, \left\{ u_{j}(\tau) \right\}^{t-1}_{\tau = 0} \mid \forall j \in \mathcal{J}_{i} \right),
\end{align} 
where $\pi_{t} \in \Pi$ and $\Pi$ denotes the set of all admissible control policies. To this end, we now formally define the distributed minimax adaptive control problem statement.

\begin{problem} \label{problem_1}
Let $\gamma > 0$ and $T \in \mathbb{N}$ denote the given time horizon. With the uncertainty set $\mathbf{A_{i}}$ for every node $i \in \mathcal{V}$ given by \eqref{eqn_node_uncertainty}, find a distributed control policy $\pi_{t} \in \Pi, \forall t \in [0,T]$ in the lines of \eqref{eqn_control_policy} to solve
\begin{align}
\label{eqn_cost}
\inf_{\pi \in \Pi} \underbrace{\sup_{a_{i} \in \mathbf{A_{i}}, w_{i}, T} \sum^{T}_{\tau=0} \left( \abs{x_{i}(\tau)}^{2} + \abs{u_{i}(\tau)}^{2} - \gamma^{2} \abs{w_{i}(\tau)}^{2} \right)}_{J^{\pi}_{i}}.
\end{align}
\end{problem}
It is evident from \eqref{eqn_cost} that there is a dynamic game being played between the control input $u_{i}$ (minimising player) and the adversaries\footnote{Note that the supremum with respect to $T$ in \eqref{eqn_cost} can be achieved by simply letting $T \rightarrow \infty$.} namely the disturbance $w_{i}$ and the uncertain parameter $a_{i} \in \mathbf{A_{i}}$ (maximising players). 
Specifically, we are interested in obtaining a condition on the $\ell_{2}$ gain of the network system $\gamma$ such that the dynamic game associated with the Problem \ref{problem_1} has a finite value (that is, $J^{\pi}_{i} < \infty$).

\subsection{Avoiding the Perils of Combinatorial Setting}

For each node $i \in \mathcal{V}$, denote the minimum and maximum value of $a_{i}$ respectively as $\underline{a}_{i} := \min\{ \mathbf{A_{i}} \}$, and $\overline{a}_{i} := \max\{ \mathbf{A_{i}} \}$. Similarly, denote the network level minimum and maximum values as $\overline{\mathbf{a}} = \max\{\overline{a}_{1}, \dots, \overline{a}_{N}\}$ and $\underline{\mathbf{a}} = \min\{ \underline{a}_{1}, \dots, \underline{a}_{N} \}$ respectively. Further, denote $\underline{A} := \mathrm{diag}(\underline{a}_{1}, \dots, \underline{a}_{N})$ and $\overline{A} := \mathrm{diag}(\overline{a}_{1}, \dots, \overline{a}_{N})$ respectively. Note that $\forall p \in \left \llbracket F \right \rrbracket$, we see that 
\begin{align} \label{eqn_A_matrix_bounds}
\underline{\mathbf{a}} I_{N} \preceq \underline{A} \preceq A_{p} \preceq \overline{A} \preceq \overline{\mathbf{a}} I_{N}.    
\end{align}

\begin{proposition} \label{proposition_matrix_bnds}
Given \eqref{eqn_A_matrix_bounds}, we observe that $\forall p \in \left \llbracket F \right \rrbracket$,
{
\footnotesize
\begin{align}
&\frac{1}{(\underline{\mathbf{a}} - 1)} I_{N} \succeq (\underline{A} - I)^{-1} \succeq (A_{p} - I)^{-1} \succeq (\overline{A} - I)^{-1} \succeq \frac{1}{(\overline{\mathbf{a}} - 1)} I_{N} \label{eqn_matrix_bounds_1}\\
&\frac{1}{(1 - \underline{\mathbf{a}})} I_{N} \preceq (I - \underline{A})^{-1} \preceq (I - A_{p})^{-1} \preceq (I - \overline{A})^{-1} \preceq \frac{1}{(1 - \overline{\mathbf{a}})} I_{N} \label{eqn_matrix_bounds_2}
\end{align}
}
\end{proposition}
\begin{proof}
From \eqref{eqn_A_matrix_bounds}, we see that
\begin{align*}
&(\underline{\mathbf{a}} - 1) I_{N} \preceq \underline{A} - I \preceq A_{p} - I \preceq \overline{A} - I \preceq (\overline{\mathbf{a}} - 1) I_{N}. \\
&(1-\underline{\mathbf{a}}) I_{N} \succeq I - \underline{A} \succeq I - A_{p} \succeq I - \overline{A} \succeq (1 - \overline{\mathbf{a}}) I_{N}.
\end{align*}
Then, by taking the inverse of both, we get \eqref{eqn_matrix_bounds_1} and \eqref{eqn_matrix_bounds_2}.
\end{proof}

\section{Distributed Minimax Adaptive Controller With Lower \& Upper $\ell_{2}$ Gain Bounds} \label{sec:analysis}
With the uncertainty given by \eqref{eqn_node_uncertainty}, the corresponding $\ell_{2}$ gain of the optimal distributed minimax adaptive controller namely $\gamma^{\dagger}$ is not yet known. Note that Problem \ref{problem_1} has a finite solution if and only if $\gamma \geq \gamma^{\dagger}$. First, we present the proposed distributed minimax adaptive control algorithm followed by lower and upper bounds for the $\ell_{2}$ gain $\gamma^{\dagger}$ corresponding to the proposed controller solving problem \ref{problem_1}.

\subsection{Design strategy for a sub-optimal distributed minimax adaptive control algorithm}
We describe a sub-optimal distributed minimax adaptive control algorithm given the set $\mathbf{A_{i}}$ for every node $i \in \mathcal{V}$. it is possible to establish an one-to-one correspondence between Problem \ref{problem_1} and the centralised minimax adaptive control problem setting given in \cite{rantzer2021minimax}. However, as the number of nodes $N$ grow, computing a centralized controller becomes expensive. To hedge against the uncertainty in the system matrix $A$, one can consider collecting historical data. Given the large-scale network setting, the highlighting point of our approach is that it is not necessary for a node $i \in \mathcal{V}$ to collect the history for the entire network. Each node $i \in \mathcal{V}$ in the network has access to only local information from its neighbors at a time. Hence, the process of hedging against the uncertainty prevailing over its dynamics $a_{i} \in \mathbf{A_{i}}$ involves collecting historical data from only its local neighbors $\mathcal{N}_{i}$ to arrive at its own control action by finding a model that best describes the disturbance trajectory up to that time. Note that the disturbance at time $t \in \mathbb{N}$ can be inferred from \eqref{eqn_node_dynamics} as
\begin{align} \label{eqn_node_disturbance_dynamics}
    w_{i}(t) = a_{i} x_{i}(t) + b \sum_{j \in \mathcal{N}_{i}} (u_{i}(t) - u_{j}(t)) - x_{i}(t+1).
\end{align}
Let us denote the control vector from node $i \in \mathcal{V}$ corresponding to all its neighbors $\mathcal{N}_{i}$ as
\begin{align}
u_{\mathcal{N}_{i}}(t)
=
\begin{bmatrix}
u_{i}(t) - u_{j_{1}}(t) \\ u_{i}(t) - u_{j_{2}}(t) \\ \vdots \\ u_{i}(t) - u_{j_{d_{i}}}(t)     
\end{bmatrix}.
\end{align}
Note that for the distributed control policy to fit the description given in \eqref{eqn_control_policy}, every node $i \in \mathcal{V}$ collects only the local neighbor information data in the form of sample covariance matrix with $t \in \mathbb{N}$ and $Z^{(i)}(0) = 0$ as
\begin{subequations}
\label{eqn_dmac_design_strategy}
\begin{align}
    x_{i}(t+1) &= v_{i}(t), \\
    Z^{(i)}(t+1) &= Z^{(i)}(t) + \begin{bmatrix} -v_{i}(t) \\ x_{i}(t) \\ u_{\mathcal{N}_{i}}(t) \end{bmatrix} \begin{bmatrix} -v_{i}(t) \\ x_{i}(t) \\ u_{\mathcal{N}_{i}}(t) \end{bmatrix}^{\top}.
\end{align}
\end{subequations}
Note that with the above construction, 
\begin{align} \label{eqn_dist_trajectory}
\norm{\begin{bmatrix} 1 & a_{i} & b \mathbf{1}^{\top}_{d_{i}}    \end{bmatrix}^{\top}}^{2}_{Z^{(i)}(t)}   
= 
\sum^{t-1}_{\tau=0} \abs{w_{i}(\tau)}^{2}.
\end{align}
Our approach is to select a model at every point in time that best describes the disturbance trajectory given by \eqref{eqn_dist_trajectory}. We do \emph{not claim} here that our distributed controller design strategy based on the disturbance estimation using historical data is an optimal one. 

\subsection{Lower bound for $\ell_{2}$ gain}
It is known that the distributed minimax adaptive controller can never do better than the distributed $\mathcal{H}_{\infty}$ optimal controller as the latter operates with the knowledge of the true system dynamics. This implies that we can obtain lower bound for the $\ell_{2}$ gain corresponding to the distributed minimax adaptive controller using the associated $\ell_{2}$ gain of the distributed $\mathcal{H}_{\infty}$ optimal controller. The following lemma formally establishes this fact. 

\begin{lemma} \label{lemma_l2_gain_lower_bound}
Given the uncertainty set $\mathbf{A_{i}}$ for every node $i \in \mathcal{V}$ described by \eqref{eqn_node_uncertainty}, 
\begin{align} \label{eqn_gamma_dagger_lower_bound}
    \gamma^{\dagger} \geq \underbrace{\norm{\left((\overline{A}-I)^{2} + BB^{\top} \right)^{-1}}^{\frac{1}{2}}}_{:= \underline{\gamma}^{\dagger}}.
\end{align}
\end{lemma}
\begin{proof}
Note that given a known $(A,B)$ pair, Theorem 1 of \cite{lidstrom2017h} gives an explicit expression for the controller gain matrix and the corresponding $\ell_{2}$ gain achieved by the controller from the disturbance to the error. That is, given $(A, B)$ matrices, the respective $\ell_{2}$ gain $\gamma^{\star}$ is achieved by the controller\footnote{The corresponding control input for each node $i \in \mathcal{V}$ is $u^{\star}_{i}(t) = \frac{b x_{i}(t)}{a_{i} - 1}$.} given by \eqref{eqn_opt_distributed_Hinf_control} and $\gamma^{\star}$ is given by 
\begin{align} \label{eqn_Hinfty_l2_gain}
\gamma^{\star} := \norm{\left((A - I)^{2} + BB^{\top}\right)^{-1}}^{\frac{1}{2}}.    
\end{align}
Since $A$ is diagonal and Schur stable, we observe that the lower bound for the $\ell_{2}$ gain $\gamma^{\dagger}$ is achieved \footnote{One can substitute $A = \overline{\mathbf{a}} I_{N}$ too to get a lower bound.} when $A  = \overline{A}$ and hence the result follows. 
\end{proof}
\begin{lemma} \label{lemma_Riccati_Aide}
Given the condition \eqref{eqn_compact_communication_constraint} $\forall p \in \left \llbracket F \right \rrbracket$, we can infer from Lemma \ref{lemma_l2_gain_lower_bound} that
\begin{align} \label{eqn_lemma_condition}
    (I - A_{p}) \succ \gamma^{-2} I.
\end{align}
\end{lemma}
\begin{proof}
Notice $\forall p \in \left \llbracket F \right \rrbracket$, \eqref{eqn_gamma_dagger_lower_bound} implies that
\begin{subequations}
\begin{align} 
    &(A_{p} - I)^{2} + BB^{\top} \succeq \gamma^{-2} I \\
    \iff &BB^{\top} \succeq \gamma^{-2} I - (A_{p} - I)^{2} \label{eqn_lemma_assumption}
\end{align}
\end{subequations}
Then $\forall p \in \left \llbracket F \right \rrbracket$, rewriting \eqref{eqn_compact_communication_constraint}, we get
\begin{align} \label{eqn_lemma_intermediate_step1}
(I - A_{p}) A_{p} - B B^{\top} \succ 0.
\end{align}
Using \eqref{eqn_lemma_assumption} in \eqref{eqn_lemma_intermediate_step1}, we get,
\begin{align*}
&(I - A_{p}) A_{p} + (I - A_{p})^{2} - \gamma^{-2} I \succ 0, \\
\iff
&(I - A_{p}) (A_{p} + I - A_{p}) - \gamma^{-2} I \succ 0
\\
\iff
&(I - A_{p}) \succ \gamma^{-2} I.
\end{align*}
\end{proof}
A solution satisfying the $\mathcal{H}_{\infty}$ Riccati inequality corresponding to the controller \eqref{eqn_opt_distributed_Hinf_control} is shown below.

\begin{lemma} \label{lemma_Hinf_riccati_ineq_solution}
Let $\gamma > 0$. Then, $\forall p \in \left \llbracket F \right \rrbracket$, the matrix $P_{p} = (I - A_{p})^{-1}$ associated with the controller gain $K_{p} := B^{\top} \left(A_p - I \right)^{-1}$ satisfies the $\mathcal{H}_{\infty}$ Riccati inequality 
\begin{align} \label{eqn_H_inf_riccati_inequality}
    P_{p} &\succeq I + K^{\top}_{p} \, K_{p} \nonumber \\
          &+ (A_{p} + B K_{p})^{\top} \left( P^{-1}_{p} - \gamma^{-2} I \right)^{-1} (A_{p} + B K_{p}).
\end{align}
\end{lemma}

\begin{proof}
Given \eqref{eqn_compact_communication_constraint}, we see that 
\begin{align*}
    0 &\preceq
    \begin{bmatrix}
        A_{p} - A^{2}_{p} - BB^{\top} & A_{p} - A^{2}_{p} - BB^{\top} \\
        A_{p} - A^{2}_{p} - BB^{\top} & A_{p} - A^{2}_{p} - BB^{\top}
    \end{bmatrix} \\
    &\preceq
    \begin{bmatrix}
        I - A_{p} - \gamma^{-2} I & A_{p} - A^{2}_{p} - BB^{\top} \\
        A_{p} - A^{2}_{p} - BB^{\top} & A_{p} - A^{2}_{p} - BB^{\top}
    \end{bmatrix},     
\end{align*}    
where we used Lemma \ref{lemma_Riccati_Aide} to get the second inequality. Now, multiplying the above inequality from left and right by $\begin{bmatrix} I & 0 \\ 0 & (A_{p} - I)^{-1} \end{bmatrix}$, and substituting $K_{p} := B^{\top} \left(A_p - I \right)^{-1}, P_{p} = (I - A_{p})^{-1}$, we get
\begin{align*}
\begin{bmatrix}
    P^{-1}_{p} - \gamma^{-2} I & -(A_{p} + BK_{p})^{\top} \\
    -(A_{p}+ BK_{p}) & P_{p} - I - K^{\top}_{p} K_{p}
\end{bmatrix} \succeq 0.
\end{align*}    
The result \eqref{eqn_H_inf_riccati_inequality} follows by applying Schur complement.
\end{proof}

\begin{theorem} \label{theorem_1}
When $A = \overline{\mathbf{a}}I_{N}$ in \eqref{eqn_compact_dynamics}, Lemma \ref{lemma_Hinf_riccati_ineq_solution} implies that
\begin{align} \label{eqn_hinf_gamma_bound}
    \gamma \geq \norm{\mathbf{G}_{\infty} ((1 - \overline{\mathbf{a}})\mathbf{G}_{\infty} - \mathbf{F}_{\infty})^{-1}}^{\frac{1}{2}}, 
\end{align}
where $\mathbf{G}_{\infty} = \frac{1}{(1 - \overline{\mathbf{a}})^{2}} ((1 - \overline{\mathbf{a}})\overline{\mathbf{a}}I - BB^{\top})$, and $\mathbf{F}_{\infty} = \left( \overline{\mathbf{a}} I + \frac{1}{\overline{\mathbf{a}} - 1} BB^{\top} \right)^{\top} \left( \overline{\mathbf{a}} I + \frac{1}{\overline{\mathbf{a}} - 1} BB^{\top} \right)$.
\end{theorem}
\begin{proof}
From Lemma \ref{lemma_Hinf_riccati_ineq_solution}, when $A = \overline{\mathbf{a}}I_{N}$, we see that $P_{p} = \frac{1}{1 - \overline{\mathbf{a}}} I_{N}$ and $K_{p} = \frac{1}{\overline{\mathbf{a}} - 1} B^{\top}$. Using these values in \eqref{eqn_H_inf_riccati_inequality} and solving for $\gamma$, we get \eqref{eqn_hinf_gamma_bound} as it is the condition required for the associated dynamic game to have a finite value.   
\end{proof}

\noindent We now present the proposed sub-optimal distributed minimax adaptive control algorithm using the strategy given by \eqref{eqn_dmac_design_strategy} and also give an upper bound $\overline{\gamma}^{\dagger}$ for the $\ell_{2}$ gain. That is, our formulation provides a family of (sub-optimal) distributed minimax control policy parameterized by $\gamma$, which are guaranteed to exist $\forall \gamma > \overline{\gamma}^{\dagger} \geq \gamma^{\dagger}$. 


\subsection{Distributed implementation for minimax adaptive control algorithm and its $\ell_{2}$ gain upper bound.}
\label{subsec:analysis}

\begin{theorem} \label{theorem_2}
Let the uncertainty set $\mathbf{A_{i}}$ for every node $i \in \mathcal{V}$ be given by \eqref{eqn_node_uncertainty}. Then, 
for every node $i \in \mathcal{V}$, the distributed control policy that solves Problem \ref{problem_1} is of the form
\begin{subequations}
\label{eqn_distributed_minimax_control}
\begin{align} 
    u_{i}(t) &= \frac{b x_{i}(t)}{a^{\dagger}_{i}(t) - 1}, \quad \text{with } \\
    a^{\dagger}_{i}(t) &= \argmin_{a_{i} \in \mathbf{A}_{i}} \norm{\begin{bmatrix}
    1 & a_{i} & b \mathbf{1}^{\top}_{d_{i}}    
    \end{bmatrix}^{\top}}^{2}_{Z^{(i)}(t)}.
\end{align}    
\end{subequations}
Further, the dynamic game associated with the Problem \ref{problem_1} under the control policy given by \eqref{eqn_distributed_minimax_control} will have a finite value (i.e., $J^{\pi}_{i} < \infty$) if $\beta := \gamma^{2}$ satisfies
\begin{subequations}
\label{eqn_gamma_dagger_upper_bound}
\begin{align} 
&f_{1} \beta^{3} + f_{2} \beta^{2} + f_{3} \beta + f_{4} 
\geq 0, \quad \text{where}, \label{eqn_beta_cubic_polynomial}\\
&f_{1} 
=
\frac{(1-\overline{\mathbf{a}})(\overline{\mathbf{a}} - \underline{\mathbf{a}})^{2}}{8}, \\
&f_{2}
= \frac{-2\overline{\mathbf{a}}^{3} +4\overline{\mathbf{a}}^{2} -2\underline{\mathbf{a}}^{2} +4\overline{\mathbf{a}} \underline{\mathbf{a}} - 2\overline{\mathbf{a}} -2\underline{\mathbf{a}} }{4}, \\
&f_{3}
=
\frac{4\overline{\mathbf{a}}^{3} -14\overline{\mathbf{a}}^{2} + 16 \overline{\mathbf{a}} - 4\underline{\mathbf{a}} -18 }{4(1 - \overline{\mathbf{a}})} \\
&f_{4}
=
\frac{-1}{(\overline{\mathbf{a}} - 1)^{2}}.
\end{align}
\end{subequations}
\end{theorem}

\begin{proof}
To ensure that the dynamic game associated with Problem \ref{problem_1} has a finite value, we invoke Theorem 3 from \cite{rantzer2021minimax}. Note that, for all $x \in \mathbb{R}^{N}$, $k,l,p \in \left \llbracket F \right \rrbracket$, except if $k \neq l = p$ and $A^{cl}_{kp} = A_{k} + B K_{p}$, the following linear matrix inequality given in Theorem 3 from \cite{rantzer2021minimax}
\begin{equation} \label{eqn_Anders_LMI}
\begin{aligned}
\footnotesize
\abs{x}^{2}_{P_{lp}} \geq \abs{x}^{2}_{Q} + \abs{K_{p}x}^{2}_{R} &+ \abs{(A^{cl}_{kp} + A^{cl}_{lp})x/2}^{2}_{(P_{kl}^{-1} - \gamma^{-2}I)^{-1}}\\
&- \gamma^{2} \abs{(A^{cl}_{kp} - A^{cl}_{lp})x/2}^{2}
\end{aligned}
\end{equation}
reduces to the standard $\mathcal{H}_{\infty}$ coupled Riccati inequality given in \eqref{eqn_H_inf_riccati_inequality} with respect to the pair $(A_{p}, B)$ when $k = l = p$.
Further, we also know from Lemma \ref{lemma_Hinf_riccati_ineq_solution} that $\forall p \in \left \llbracket F \right \rrbracket, P_{p} = (I - A_{p})^{-1}$ satisfies \eqref{eqn_H_inf_riccati_inequality} for the controller $K_{p} = B^{\top} (A_{p} - I)^{-1}$ given by \eqref{eqn_opt_distributed_Hinf_control}. Now, to arrive at a condition that ensures $J^{\pi}_{i} < \infty$, we let $Q = I_{N}, R = I_{E}$ and assume $A_{p} = A_{k} = \overline{\mathbf{a}} I_{N}, A_{l} = \underline{\mathbf{a}} I_{N}$, $P_{kk} = \left(I - A_{k} \right)^{-1}$ and $0 \prec P_{kl} = P_{lk} = \frac{2}{1 - \overline{\mathbf{a}}} I \prec \gamma^{2}I$ for $k,l,p \in \{1,\dots,F\}$ except if $k \neq l = p$. From \eqref{eqn_A_matrix_bounds} and Proposition \ref{proposition_matrix_bnds}, we observe that $\frac{2}{1 - \overline{\mathbf{a}}} I \succ \left(I - A_{k} \right)^{-1}, \forall k \in \left \llbracket F \right \rrbracket$ and hence our choices are valid.
Now, we substitute the above assumed values of $P_{kl}, P_{kk}$ in \eqref{eqn_Anders_LMI} to see that
\begin{align} \label{thm1_proof_interstep_1}
    &\mathbf{G} 
    \succeq 
    \frac{1}{4\left( \frac{1 - \overline{\mathbf{a}}}{2} - \frac{1}{\gamma^{2}} \right)} 
    \mathbf{F} - \frac{\gamma^{2}}{4} \mathbf{H},
\end{align}
where $\mathbf{F} = \left( (\overline{\mathbf{a}} + \underline{\mathbf{a}})I + \frac{2}{1 - \overline{\mathbf{a}}} BB^{\top}
\right)^{\top} \left( (\overline{\mathbf{a}} + \underline{\mathbf{a}})I + \frac{2}{1 - \overline{\mathbf{a}}} BB^{\top}
\right)$, $\mathbf{G} = \frac{1 + \overline{\mathbf{a}}}{1 - \overline{\mathbf{a}}} I - \frac{1}{(\overline{\mathbf{a}} - 1)^{2}} B B^{\top}$, and $\mathbf{H} = (\overline{\mathbf{a}} - \underline{\mathbf{a}})^{2}I$. Now, let $\beta = \gamma^{2}$. Applying \eqref{eqn_lemma_assumption} in \eqref{thm1_proof_interstep_1} and solving for $\beta$, we arrive at the cubic polynomial in $\beta$ given by \eqref{eqn_beta_cubic_polynomial} (details are given in the appendix of \cite{venkat_dmac}). Then, condition on $\gamma$ to get a finite game value can be readily obtained through the condition on $\sqrt{\beta_{\mathrm{min}}}$, where $\beta_{\mathrm{min}}$ is the minimum positive root of the cubic polynomial \eqref{eqn_beta_cubic_polynomial}. The reasoning behind a finite value for the associated dynamic game is that the control law given by \eqref{eqn_distributed_minimax_control} indicates that every node $i \in \mathcal{V}$ selects the model that best describes the disturbance trajectory modelled using the collected history $(Z^{(i)}(t))$ in a least-square sense and then employs the corresponding optimal distributed $\mathcal{H}_{\infty}$ control law at every time step $t$ by taking the certainty equivalence principle. Then, invoking Theorem 3 from \cite{rantzer2021minimax} along with \eqref{eqn_gamma_dagger_upper_bound}, we are ensured that the control law in the distributed form for every node $i \in \mathcal{V}$ given by \eqref{eqn_distributed_minimax_control} solves Problem \ref{problem_1} with a finite value for dynamic game given in Problem \ref{problem_1} meaning that, $J^{\pi}_{i} < \infty$ and this completes the proof.
\end{proof}

\textbf{Remarks:} Observe that the right hand side of \eqref{eqn_gamma_dagger_upper_bound} is an upper bound for the $\gamma^{\dagger}$ and the dynamic game associated with Problem \ref{problem_1} will have a finite value if and only if $\gamma$ exceeds that quantity. Different bounds are possible to obtain with other choices for matrices $P_{kl} = P_{lk} \succeq P_{kk}$ while invoking Theorem 3 of \cite{rantzer2021minimax} and further \eqref{eqn_lemma_assumption} introduces conservatism. Hence, we do not claim that our result is tight. 

\section{Numerical Simulation}
\label{sec:simulation}

\subsection{Simulation Setup}
To demonstrate our proposed approach, we consider a large-scale buffer network with $N = 10^{4}$ nodes where computing a centralised controller is non-trivial and expensive. We let $b = 0.1$ so that the input matrix for the whole network is simply the scaled incidence matrix, $B = bI$, with $\mathcal{I}$ being the incidence matrix of the associated tree graph. Two different models $(M = 2)$ were generated randomly for each node $i \in \mathcal{V}$ satisfying \eqref{eqn_communication_condition}. The total time horizon was set to be $T = 20$. To simulate the above system with a disturbance signal, we chose a zero mean random signal with covariance of $0.1 I_{N}$. Model number two from the set $\mathbf{A_{i}}$ was picked and fixed to be the true system model governing its dynamics for every node $i \in \mathcal{V}$ throughout the time horizon. For every node $i \in \mathcal{V}$, the distributed minimax adaptive control inputs were computed using \eqref{eqn_distributed_minimax_control} and the distributed $\mathcal{H}_{\infty}$ control inputs were computed as described in Corollary 1 of \cite{lidstrom2017h} using the true $a_{i} \in \mathbf{A_{i}}$ as $u^{\star}_{i}(t) = \frac{bx_{i}(t)}{a_{i} - 1}$.

\subsection{Results \& Discussion}
The bounds from Lemma \ref{lemma_l2_gain_lower_bound} and Theorem \ref{theorem_2} were found out to be $13.5936$ and $16.0161$ respectively. We observed that both the distributed minimax adaptive control policy and the distributed $\mathcal{H}_{\infty}$ control policy were stabilising. That is, the behaviour of the distributed minimax adaptive and the distributed $\mathcal{H}_{\infty}$ controllers get very similar after certain point in time as shown in Figure~\ref{fig_statesControlsDiff}, where the stabilised states of node $1$ is shown for illustrative purpose. The respective control inputs behaving very similar to each other is shown in Figure~\ref{fig_statesControlsDiff} where the first edge input from both the controller is shown. The time when both these policies coincide is still an open problem as we conjecture that there might exist adversarial disturbance policies for some node in the network that can keep the controller continuously guessing about its local dynamics model uncertainty. The code is made available at \url{https://gitlab.control.lth.se/regler/distributed_mac}.

\begin{figure}
    \centering
    \begin{tikzpicture}

\begin{axis}[%
width=3.014in,
height=1.92in,
at={(1.011in,1.556in)},
scale only axis,
xmin=0,
xmax=30,
xlabel style={font=\color{white!15!black}},
xlabel={Time},
ymin=0,
ymax=0.8,
axis background/.style={fill=white},
legend style={legend cell align=left, align=left, draw=white!15!black}
]
\addplot [color=mycolor1, line width=3.0pt, mark=o, mark options={solid, mycolor1}]
  table[row sep=crcr]{%
0	0\\
1	0.0868531450295933\\
2	0.0592396129136521\\
3	0.02292800642409\\
4	0.00703726531601889\\
5	0.00503407952286938\\
6	0.00236484361248401\\
7	0.00475517153609833\\
8	0.00565861426749376\\
9	0.00782849892357232\\
10	0.00126563428970902\\
11	0.0112219585042816\\
12	0.0104068680524106\\
13	0.000552262176294407\\
14	0.00163248382029672\\
15	0.00720819079010328\\
16	0.00141424536453333\\
17	0.003482126274788\\
18	0.00484369925522417\\
19	0.00718440982967506\\
20	0.000222171329609193\\
21	0.00427816915295266\\
22	0.00330925101185076\\
23	0.00167994312998794\\
24	0.00352241622042904\\
25	0.00073999741902124\\
26	0.0022877917507506\\
27	0.00079574438631666\\
28	0.00901269252196008\\
29	0.0122897150768699\\
30	0.0128545299141771\\
};
\addlegendentry{$\left \Vert x^{\dagger}_{1} - x^{\star}_{1} \right \|_{1}$: white noise $w$}

\addplot [color=mycolor2, line width=3.0pt, mark=o, mark options={solid, mycolor2}]
  table[row sep=crcr]{%
0	0.671339457501923\\
1	0.398128632026763\\
2	0.310036870303036\\
3	0.166381340638015\\
4	0.039000507659078\\
5	0.0834422165888566\\
6	0.057811668534194\\
7	0.0181057442361426\\
8	0.0574107908604506\\
9	0.013241255703171\\
10	0.120433896143466\\
11	0.0744914873243598\\
12	0.0848430154706912\\
13	0.00567931865622329\\
14	0.0637462109334223\\
15	0.0192841692286815\\
16	0.00688372173284261\\
17	0.0103133763966439\\
18	0.0108937827510889\\
19	0.0843927763740942\\
20	0.0831113094179744\\
21	0.0068669195373595\\
22	0.0134674527273368\\
23	0.016933142888331\\
24	0.0750133472681332\\
25	0.0360956098994726\\
26	0.0135630311919922\\
27	0.120882226804632\\
28	0.0631228360651968\\
29	0.0658821454257002\\
};
\addlegendentry{$\left \Vert u^{\dagger}_{1} - u^{\star}_{1} \right \|_{1}$: white noise $w$}

\end{axis}
\end{tikzpicture}%
    \caption{The behaviour of the uncertain network controlled by minimax adaptive controller (with superscript $\dagger$) eventually coinciding with that of the distributed $\mathcal{H}_{\infty}$ controller (with superscript $\star$) under the effect of a random adversarial disturbance is depicted here. The difference of states from both controllers corresponding to the first node in the network and the difference of control inputs corresponding to the first edge in the network are shown here.}
    \label{fig_statesControlsDiff}
\end{figure}
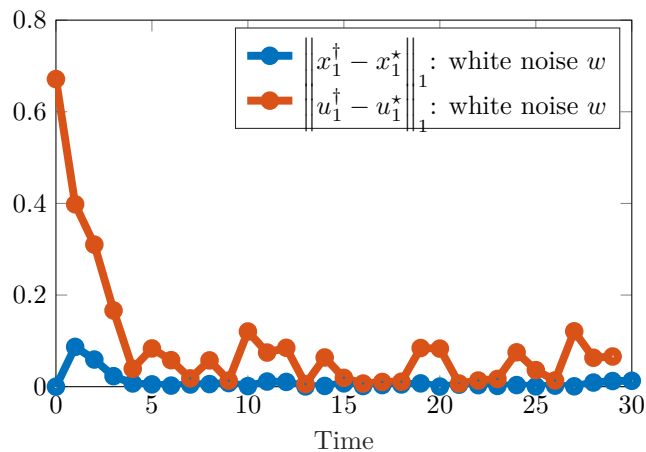


\section{Conclusion}
\label{sec:conclusion}
A distributed minimax adaptive controller for uncertain networked systems was presented in this paper. Based on the local information collected by each node in the network, the distributed minimax adaptive controller selects the best model that minimizes the disturbance trajectory hitting that node and selects the corresponding distributed $\mathcal{H}_{\infty}$ control law. Our proposed distributed implementation scales linearly with the size of the network. Both lower and upper bounds for the associated $\ell_{2}$ gain of the controller were obtained. Future work will seek to extend the framework to uncertain networked systems with output model.

\inArxiv{
\section*{Acknowledgment}
This project has received funding from the European Research Council (ERC) under the European Union’s Horizon 2020 research and innovation program under grant agreement No 834142 (Scalable Control). All authors are members of the ELLIIT Strategic Research Area in Lund University.
}


\inConf{
\bibliographystyle{IEEEtran}
\bibliography{references}
}
\inArxiv{
\bibliographystyle{tmlr}
\bibliography{references}
}

\onecolumn  
\section*{Appendix}
\subsection*{Derivation of Cubic Polynomial in $\beta$}
\noindent Recall \eqref{thm1_proof_interstep_1} from the proof of Theorem \ref{theorem_2} as 
\begin{align*} 
    &\mathbf{G} 
    \succeq 
    \frac{1}{4\left( \frac{1 - \overline{\mathbf{a}}}{2} - \frac{1}{\gamma^{2}} \right)} 
    \mathbf{F} - \frac{\gamma^{2}}{4} \mathbf{H},
\end{align*}
where 
\begin{align*}
\mathbf{F} &= \left( (\overline{\mathbf{a}} + \underline{\mathbf{a}})I + \frac{2}{1 - \overline{\mathbf{a}}} BB^{\top} \right)^{\top} 
\left( (\overline{\mathbf{a}} + \underline{\mathbf{a}})I + \frac{2}{1 - \overline{\mathbf{a}}} BB^{\top} \right),  \\
\mathbf{G} &= \frac{1 + \overline{\mathbf{a}}}{1 - \overline{\mathbf{a}}} I - \frac{1}{(\overline{\mathbf{a}} - 1)^{2}} B B^{\top}, \\
\mathbf{H} &= (\overline{\mathbf{a}} - \underline{\mathbf{a}})^{2}I. 
\end{align*}
From \eqref{eqn_lemma_assumption}, we observe that
\begin{align*}
BB^{\top} 
&> \gamma^{-2} I - (A_{p} - I)^{2} \\
&\geq \left( \gamma^{-2} - (\overline{\mathbf{a}} - 1)^{2} \right) I.
\end{align*}
Using this in the definition of $\mathbf{F}$ and $\mathbf{G}$, we get,
\begin{align*}
&\left(
\frac{1+\overline{\mathbf{a}}}{1-\overline{\mathbf{a}}}
- \frac{1}{\gamma^{2}} 
+ (\overline{\mathbf{a}} - 1)^{2} 
+ \frac{\gamma^{2}}{4} (\overline{\mathbf{a}} - \underline{\mathbf{a}})^{2} 
- \frac{\left( \underline{\mathbf{a}} - \overline{\mathbf{a}} + 2 + \frac{2}{\gamma^{2} (\overline{\mathbf{a}} - 1)} \right)^{2}}{4 \left( \frac{1-\overline{\mathbf{a}}}{2} - \frac{1}{\gamma^{2}} \right)}
\right)
I 
\succeq 0 \\
\iff
&\left(
f_{1} \beta^{3} + f_{2} \beta^{2} + f_{3} \beta + f_{4} 
\right)
I 
\succeq 0, \\
\iff
&\left(
f_{1} \beta^{3} + f_{2} \beta^{2} + f_{3} \beta + f_{4} 
\right)
\geq 0,
\end{align*}
where $\beta = \gamma^{2}$ and
\begin{align*}
&f_{1} 
=
\frac{(1-\overline{\mathbf{a}})(\overline{\mathbf{a}} - \underline{\mathbf{a}})^{2}}{8}, \\
&f_{2}
= \frac{-2\overline{\mathbf{a}}^{3} +4\overline{\mathbf{a}}^{2} -2\underline{\mathbf{a}}^{2} +4\overline{\mathbf{a}} \underline{\mathbf{a}} - 2\overline{\mathbf{a}} -2\underline{\mathbf{a}} }{4}, \\
&f_{3}
=
\frac{4\overline{\mathbf{a}}^{3} -14\overline{\mathbf{a}}^{2} + 16 \overline{\mathbf{a}} - 4\underline{\mathbf{a}} -18 }{4(1 - \overline{\mathbf{a}})} \\
&f_{4}
=
\frac{-1}{(\overline{\mathbf{a}} - 1)^{2}}.
\end{align*}
The solutions $\beta_{1}, \beta_{2}, \beta_{3}$ satisfying the cubic polynomial implies that $\gamma_{1} = \sqrt{\beta_{1}}, \gamma_{2} = \sqrt{\beta_{2}}, \gamma_{3} = \sqrt{\beta_{3}}$ also satisfy it. Since we look for the positive yet tight $\gamma$, it corresponds to $\gamma = \min \{\gamma_{1}, \gamma_{2}, \gamma_{3}\}$ such that $\gamma > 0$. 


\end{document}